\documentclass{llncs}

\usepackage{amsmath,amssymb}
\usepackage{mathptmx}
\usepackage{stmaryrd}

\usepackage{hyperref}

\newcommand{\up}{\uparrow}
\newcommand{\down}{\downarrow}

\renewcommand{\up}[1][{\hspace{-0.055555em}}]{{\uparrow_{#1}}}
\renewcommand{\down}[1][{\hspace{-0.055555em}}]{{\downarrow_{#1}}}

\newcommand{\tu}[1] {\langle #1 \rangle}

\newcommand{\class}[2][\sim]{\llbracket #2 \rrbracket_{#1}}
\newcommand{\interval}[1]{\llbracket #1 \rrbracket}
\newcommand{\simfg}{\sim_{\tu{f,g}}}

\def\deg#1#2{{}^{#1\!}/#2}
\def\Deg#1#2{{}^{#1\!}/#2}

\def\from{\leftarrow}

\def\lat{\mathbf L}
\def\latU{\mathbf U}
\def\latV{\mathbf V}
\def\cut#1#2{{}^{#1}{#2}}

\def\Lcone{\mathcal L}
\def\Ucone{\mathcal U}

\DeclareMathOperator{\igal}{IGal}
\DeclareMathOperator{\eigal}{EIGal}
\DeclareMathOperator{\ctol}{CTol}
\DeclareMathOperator{\Fix}{Fix}

\def\Fixfg{\Fix_{\tu{f,g}}}

\begin{document}

\title{Complete relations on fuzzy complete lattices\thanks{This paper is an extended and thoroughly rewritten version of a part of a paper presented at CLA 2011.
We also fix some inaccuracies that appeared in the original paper.}
\thanks{The authors acknowledge support  
  by the ESF project No. CZ.1.07/2.3.00/20.0059, the project
  is co-financed by the European Social Fund and the state 
  budget of the Czech Republic.}}
\subtitle{(Preprint)}

\author{Jan Konecny \and Michal Krupka}

\institute{Data Analysis and Modeling Lab,\\
Dept. of Computer Science, Palack\'y University, Olomouc\\
17. listopadu 12, Olomouc, Czech Republic \\
\email{jan.konecny@upol.cz} \\
\email{michal.krupka@upol.cz}}


\maketitle

\begin{abstract}
We generalize the notion of complete binary relation on complete lattice to residuated lattice valued ordered sets and show its properties. Then we focus on complete fuzzy tolerances on fuzzy complete lattices and prove they are in one-to-one correspondence with extensive isotone Galois connections. Finally, we prove that fuzzy complete lattice, factorized by a complete fuzzy tolerance, is again a fuzzy complete lattice.
\end{abstract}

\section{Introduction}
In classical algebra, a complete relation on a complete lattice is a relation which preserves arbitrary infima and suprema. For instance, a binary relation $\sim$ on a complete lattice $U$ is complete, if for each system $\{\tu{u_i,v_i}\}_{i\in I}$ of pairs of elements from $U$, $u_i\sim v_i$ for each $i\in I$ implies $\bigwedge_{i\in I} u_i\sim\bigwedge_{i\in I} v_i$ and $\bigvee_{i\in I} u_i\sim\bigvee_{i\in I} v_i$.

One of the goals of this paper is to define a notion of complete relation for fuzzy sets. That is, we need to state an appropriate condition for completeness of a  fuzzy relation on a set, possessing an appropriate structure of a complete lattice in fuzzy sense. However, the above definition cannot be used as is. 

As it turns out, there is an equivalent condition to that of completeness of a relation on a complete lattice, that involves extending relations between sets to relations between power sets (i.e. sets of all subsets). This situation is known from the theory of so called \emph{power algebras} \cite{Bri:Ps}, which offers a natural way to extend a binary relation $R$ on a set $X$ to a binary relation $R^+$ on the power set $2^X$.

This extension allows us formulate the following equivalent condition for completeness of binary relations: a binary relation $\sim$ on a complete lattice $\latU$ is complete, if and only if for any two subsets $V_1,V_2$ in $\latU$, $V_1 \sim^+ V_2$ implies $\bigwedge V_1\sim\bigwedge V_2$ and $\bigvee V_1\sim\bigvee V_2$.

In \cite{Geo:Fps}, Georgescu extended the theory of power algebras to a fuzzy setting. He shows a way of extending any fuzzy $n$-ary relation $R$ on a set $X$ to a fuzzy $n$-ary relation on the set of all fuzzy sets in $X$. In this paper, we use these results to define a notion of a~complete binary fuzzy relation on a complete fuzzy lattice.

As a general framework, we use $\lat$-valued fuzzy sets, where $\lat$ is a complete residuated lattice, thus covering $[0,1]$-valued fuzzy sets with arbitrary left-continuous t-norm on $[0,1]$ as a special case. Under this framework, we use a notion of $\lat$-ordered set, which is, basically, a set with an $\lat$-relation satisfying requirements of reflexivity, antisymmetry and transitivity. A complete fuzzy lattice, or, more precisely, a completely lattice $\lat$-ordered set, is then an $\lat$-ordered set whose each $\lat$-subset has a (properly defined) infimum and supremum.

$\lat$-valued fuzzy sets, completely lattice $\lat$-ordered sets and other basic notions from the fuzzy set theory (e.g. isotone $\lat$-Galois connections and $\lat$-closure and $\lat$-interior operators) are introduced in Sec.~\ref{Sec:prelim}. In this section, we also prove some basic new results we need in subsequent parts of the paper, namely some properties on isotone $\lat$-Galois connections.

Sec.~\ref{Sec:power} is devoted to some basic parts of the Georgescu's theory of fuzzy power structures and its applications to $\lat$-ordered sets. We start with recalling the notion of power binary $\lat$-relations and their basic properties and then we prove some results on power relations of $\lat$-orders. Section \ref{Sec:compl} contains our definition
of complete binary $\lat$-relation on completely lattice $\lat$-ordered set. We also prove some basic properties of complete $\lat$-relations.

In the main part of the paper, Section \ref{Sec:tol}, we focus on complete fuzzy  tolerances. A (crisp) tolerance on a set is a reflexive and symmetric binary relation. A block of a tolerance is a set whose elements are pairwise related. A maximal block is a block which is maximal w.r.t. set inclusion. The set of all maximal blocks of a tolerance is called the factor set. One of basic results on tolerances on complete lattices is that complete lattices can be factorized by complete tolerances \cite{Cze:Flt,Wille:Ctrcl}. That is, there can be introduced in a natural way an ordering on the set of all maximal blocks of a complete tolerance, such that the factor set, together with this ordering, is again a complete lattice.

We show that the same holds for complete $\lat$-tolerances on completely lattice $\lat$-ordered sets. More precisely, we use the usual definition of fuzzy tolerance and corresponding factor set and introduce an $\lat$-order on the factor set of completely lattice $\lat$-ordered set by a complete $\lat$-tolerance, such that the new $\lat$-order is again a complete lattice $\lat$-order.

To prove this main result, we investigate deeply properties of complete $\lat$-tolerances. We use similar techniques to those used in classical ordered sets. However, we also introduce a result that is new even in the classical case: we show that complete fuzzy tolerances are in one-to-one correspondence with so-called extensive isotone fuzzy Galois connections.

Note that factorization of complete lattices, either in ordinary or fuzzy setting, has been studied in the past \cite{Wille:Ctrcl,GanWi:FCA,Belohlavek:2007aa,Bel:FRS,Kr:Frl} as it is useful for reducing dimensionality of concept lattices. This paper can be viewed as a contribution to this area. 


\section{Preliminaries}\label{Sec:prelim}
\subsection{Residuated lattices and fuzzy sets}
A \emph{complete residuated lattice} \cite{Bel:FRS,Haj:MFL,WaDi:Rl} is a
structure $\mathbf{L}=\langle
L,\wedge,\vee,\otimes,\rightarrow,0,1\rangle$ such that
\begin{enumerate}\parskip=2pt
\item[(i)]
$\langle L,\wedge,\vee,0,1\rangle$ is a complete
lattice, i.e. a partially ordered set in which arbitrary infima
and suprema exist, $0=\bigwedge L$, $1 = \bigvee L$; 
\item[(ii)]
$\langle L,\otimes,1\rangle$ is a commutative monoid, i.e. $\otimes$
is a binary operation which is commutative, associative,
and $a \otimes 1 = a$ for each $a \in L$;
\item[(iii)]
$\otimes$ and $\rightarrow$ satisfy adjointness, i.e.  $a\otimes b\leq c$
if{}f $a\leq b\rightarrow c$.
\end{enumerate}
The partial order
of $\mathbf{L}$ is denoted by $\leq$.
Throughout the paper,
$\mathbf{L}$ denotes an arbitrary complete residuated lattice.

Elements of $L$ are called truth degrees. $\otimes$ and
$\rightarrow$ are (truth functions of) ``fuzzy conjunction'' and
``fuzzy implication''.

Common examples of complete residuated lattices include those defined
on $[0,1]$, (i.e. $L=[0,1]$), $\wedge$ being minimum,
$\vee$ maximum, $\otimes$ being a left-continuous
\mbox{t-norm} with the corresponding $\rightarrow$. 

The three most important pairs
of adjoint operations on the unit interval are
\medskip

\begin{tabular}{ll}
\L ukasiewicz: &
$\begin{array}{rcl}
  a \otimes b&=&\mathrm{max} (a + b - 1, 0) \\
  a \rightarrow b&=&\mathrm{min} (1 - a + b, 1)
\end{array}$\\[15pt]
G\"odel: &
$\begin{array}{rcl}
  a \otimes b&=& \mathrm{min} (a, b) \\
  a \rightarrow b&=&
  \left\{
  \begin{array}{ll}
    1 & a \le b \\
    b & \mbox{otherwise}
  \end{array}
  \right. 
\end{array}$
\\[15pt]
Goguen (product): &
$\begin{array}{rcl}
  a \otimes b&=&a \cdot b \\
  a \rightarrow b&=& \left\{
  \begin{array}{ll}
    1 & a \le b \\
    \frac{b}{a} & \mbox{otherwise}
  \end{array}
  \right.
\end{array}$
\end{tabular}
\medskip


An \emph{$\mathbf{L}$-set} (or \emph{fuzzy set}) $A$ in a universe set $X$ is a
mapping assigning to each $x \in X$ some truth degree $A(x)\in L$. The set of all
$\mathbf{L}$-sets in a universe $X$ is denoted 
$L^X$.

The operations with $\mathbf{L}$-sets are defined elementwise. For instance,
\emph{the intersection of $\mathbf{L}$-sets $A,B \in L^X$} is
an $\mathbf{L}$-set $A \cap B$ in $X$ such that
$(A \cap B)(x) = A(x) \wedge B(x)$ for each $x \in X$, etc.
An $\mathbf{L}$-set $A \in L^X$ is also denoted $\{\Deg{A(x)}{x}\,|\,x\in X \}$.
If for all $y \in X$ distinct from $x_1, x_2, \dots, x_n$ we have $A(y) = 0$, we also write
$\{ \Deg{A(x_1)}{x_1}, \Deg{A(x_2)}{x_1}, \dots, \Deg{A(x_n)}{x_n} \}$.

Binary $\mathbf{L}$-relations (binary fuzzy relations) between $X$ and $Y$
can be thought of as $\mathbf{L}$-sets in the universe $X \times Y$.
That is, a \emph{binary $\mathbf{L}$-relation $I\in L^{X\times Y}$ between
a set $X$ and a set $Y$} is a mapping assigning to each $x\in X$ and each $y\in
Y$ a truth degree $I(x,y)\in L$ (a~degree to which $x$ and $y$ are related by
$I$). The \emph{inverse relation $I^{-1}$ to the $\lat$-relation $I$} is an $\lat$-set in $Y\times X$ and is defined by $I^{-1}(y,x)=I(x,y)$.

The \emph{composition $R\circ T$ of binary $\lat$-relations $R\in L^{X\times Y}$ and $T\in L^{Y\times Z}$} \cite{KoBa:Rpaip} is a binary $\lat$-relation between $X$ and $Z$ defined by
\begin{align}
(R\circ T)(x,z) &= \bigvee_{y\in Y}R(x,y)\otimes T(y,z).
\end{align}
$\lat$-sets in a set $X$ can be naturally identified with binary $\lat$-relations between $\{1\}$ and $X$, resp.~$X$ and $\{1\}$. Thus, we can also consider composition of an $\lat$-sets and a binary $\lat$-relation and even composition of two $\lat$ set: for $A,A_1,A_2\in\lat^X$, $B\in \lat^Y$ and $R\in\lat^{X\times Y}$ we have
\begin{align}
(A\circ R)(y) &= \bigvee_{x\in X}A(x)\otimes R(x,y), &
(R\circ B)(x) &= \bigvee_{y\in Y}R(x,y)\otimes B(y), 
\end{align}
and
\begin{align}
A_1\circ A_2 &= \bigvee_{x\in X}A_1(x)\otimes A_2(x).
\end{align}

An $\mathbf{L}$-set $A\in L^X$ is called \emph{crisp} if $A(x)\in\{0,1\}$
for each $x\in X$. Crisp $\mathbf{L}$-sets can be identified with ordinary
sets. For a crisp 
$\mathbf{L}$-set $A$, we also write $x\in A$ for $A(x)=1$ and 
$x\not\in A$ for $A(x)=0$.
An $\mathbf{L}$-set $A\in L^X$ is called \emph{empty} (denoted by $\emptyset$)
if $A(x)=0$ for each $x\in X$. 
For $a\in L$ and $A\in L^X$, $a\otimes A\in L^X$
and $a\rightarrow A\in L^X$ are defined by
\[(a\otimes A)(x)=a\otimes A(x) \mbox{ and }
(a\rightarrow A)(x)=a\rightarrow A(x).\]

For an $\lat$-set $A\in L^X$ and $a\in L$, the \emph{$a$-cut of $A$} is a crisp subset $\cut aA\subseteq X$ such that $x\in \cut aA$ i{f}f $a\leq A(x)$. This definition applies also to binary $\lat$-relations, whose $a$-cuts are classical (crisp) binary relations. 

For a universe $X$ we
define an $\mathbf{L}$-relation of {\em graded subsethood} $L^X \times L^X \rightarrow L$ by:
\begin{align}
\label{eq:S}
S(A,B) &= \bigwedge_{x \in X} A(x) \rightarrow B(x).
\end{align}
Graded subsethood generalizes the classical subsethood relation $\subseteq$;
indeed, in the crisp case (i.e. $L=\{0,1\}$) \eqref{eq:S} becomes
$S(A,B)=1$ if{}f for each ${x \in X}: x \in A \text{ implies } y \in B$.
Note that $S$ is a binary $\mathbf{L}$-relation on $L^X$.
Described verbally, $S(A,B)$ represents a~degree to which $A$ is a subset of
$B$. In particular, we write $A \subseteq B$ if{}f $S(A,B)=1$. As a~consequence,
we have $A \subseteq B$ if{}f $A(x) \le B(x)$ for each $x \in X$. 

Further we set
\begin{align}
\label{eq:approxx}
A \approx^X B = S(A,B) \wedge S(B,A).
\end{align}
The value $A \approx^X B$ is interpreted as the degree to which the sets $A$ and $B$ are similar.

A binary $\lat$-relation $R$ on a set $X$ is called \emph{reflexive} if $R(x,x)=1$ for any $x\in X$, \emph{symmetric} if $R(x,y)=R(y,x)$ for any $x,y\in X$, 
and \emph{transitive} if $R(x,y)\otimes R(y,z)\leq R(x,z)$ for any $x,y,z\in X$. $R$ is called an \emph{$\lat$-tolerance}, if it is reflexive and symmetric, \emph{$\lat$-equivalence} if it is reflexive, symmetric and transitive. If $R$ is an $\lat$-equivalence such that for any $x,y\in X$ from $R(x,y)=1$ it follows $x=y$, then $R$ is called an \emph{$\lat$-equality} on $X$. $\lat$-equalities are often denoted by $\approx$. 
The similarity $\approx^X$ of $\lat$-sets
\eqref{eq:approxx}
is an $\lat$-equality on $L^X$.

Let $\sim$ be an $\lat$-equivalence on $X$.
We say that an $\lat$-set $A$ in $X$ is  \emph{compatible with $\sim$} (or \emph{extensional w.r.t. $\sim$}, if for any $x,x'\in X$ it holds
\begin{align}
A(x)\otimes(x\sim x')&\leq A(x').
\end{align}
A binary $\lat$-relation $R$ on $X$ is \emph{compatible with $\sim$}, if for each $x,x',y,y'\in X$,
\begin{align}
R(x,y)\otimes(x\sim x')\otimes(y\sim y')\leq R(x',y').
\end{align}

\emph{Zadeh's extension principle} \cite{Zad:Clvaar} allows extending any mapping $f\!:X\to Y$ to a mapping $f^+\!: L^X\to L^Y$ by setting for each $A\in L^X$
\begin{align}
f^+(A)(y) = \bigvee_{x\in X, f(x) = y}A(x). \label{eqn:image}
\end{align}

In the
following we use well-known properties of residuated lattices and fuzzy
structures which can be found e.g. in \cite{Bel:FRS,Haj:MFL}.


\subsection{$\lat$-ordered sets}\label{subsec:L_ord_sets}
In this section, we recall basic definitions and results of the theory of $\lat$-ordered sets. Basic references are \cite{Bel:Clofl,Bel:FRS} and the references therein. 

An \emph{$\lat$-order} on a set $U$ with an $\lat$-equality $\approx$ is a binary $\lat$-relation $\preceq$ on $U$ which is compatible with $\approx$, reflexive, transitive and satisfies
$(u\preceq v)\wedge (v\preceq u) \leq u\approx v$ for any $u,v\in U$ (\emph{antisymmetry}).
The tuple $\mathbf U = \tu{\tu{U,\approx},\preceq}$ is called an \emph{$\lat$-ordered set}.
An immediate consequence of the definition is that for any $u,v\in U$ it holds
\begin{equation}\label{eqn:antisym_eq}
u\approx v = (u\preceq v)\wedge (v\preceq u) .
\end{equation}

If $\latU=\tu{\tu{U,\approx},\preceq}$ is an $\lat$-ordered set, then the tuple $\tu{U,\cut 1{\preceq}}$, where $\cut 1{\preceq}$ is the $1$-cut of $\preceq$, is a (partially) ordered set. We sometimes write $\leq$ instead of $\cut 1{\preceq}$ and use the symbols $\wedge$, $\bigwedge$ resp. $\vee$, $\bigvee$ for denoting infima resp. suprema in $\tu{U,\cut 1{\preceq}}$.

For two $\lat$-ordered sets $\mathbf U=\tu{\tu{U,\approx_U},\preceq_U}$ and $\mathbf V = \tu{\tu{V,\approx_V},\preceq_V}$, a mapping $f\!:U\to V$ is \emph{isotone}, if $(u_1\preceq_U u_2) \leq (f(u_1)\preceq_V f(u_2))$ for any $u_1,u_2\in V$. The mapping $f$ is called an \emph{isomorphism of $\mathbf U$ and $\mathbf V$}, if it is a bijection and 
$(u_1\preceq_U u_2) = (f(u_1)\preceq_V f(u_2))$
for any $u_1,u_2\in V$. 
$\mathbf U$ and $\mathbf V$ are then called \emph{isomorphic}.

In classical theory of ordered sets, a subset $V$ of an ordered set is called a lower set, if for each element $u$ such that there is $v\in V$ satisfying $u\leq v$, it holds $u\in V$. Equivalently, for a lower set $V$ it holds: if $u\leq v$, then $v\in V$ implies $u\in V$. 

Analogously, for an $\lat$-ordered set $\latU$, an $\lat$-set 
$V\in L^U$
is called a \emph{lower set} (resp.~an \emph{upper set}), if for each $u,v\in U$ it holds
\begin{align}
u\preceq v &\leq V(v)\to  V(u) & \text{(resp.~} u\preceq v &\leq  V(u)\to  V(v)\text{).}
\end{align}
\emph{The lower} (resp.~\emph{upper}) \emph{set of an $\lat$-set $V\in L^U$} is the $\lat$-set $\down V$ (resp.~$\up V$), defined by
\begin{align}
\down V(u) &= ({\preceq}\circ V)(u)= \bigvee_{v\in U}(u\preceq v)\otimes V(v),
\\ 
\up V(u)&= (V\circ{\preceq})(u) =\bigvee_{v\in U}(v\preceq u)\otimes V(v).
\end{align}
In a similar manner we define lower and upper cone of $V\in L^U$. For any $v\in U$ we set
\begin{align}
\Lcone V(v) &= \bigwedge_{u\in U}V(u)\to (v\preceq u), &
\Ucone V(v) &= \bigwedge_{u\in U}V(u)\to (u\preceq v).
\end{align}
The right-hand side of the first equation is the degree of ``For each $u\in U$, if $u$ is in $V$, then $v$ is less than or equal to $u$'', and similarly for the second equation. Thus, $\Lcone V(v)$ ($\Ucone V(v)$) can be seen as the degree to which $v$ is less (greater) than or equal to each element of $V$, that is \emph{the degree to which $v$ is a lower} (\emph{upper}) \emph{bound of $V$}.

In the case $\Lcone V(v)=1$ (resp.~$\Ucone V(v)=1$) we say simply $v$ is a \emph{lower} (\emph{upper}) \emph{bound of $V$}.
$\Lcone V$ (resp. $\Ucone V$) is called the \emph{$\lat$-set of lower bounds} (resp.~\emph{upper bounds}) \emph{of $V$}, or  \emph{the lower cone} (resp. \emph{the upper cone}) of $V$.

If $u,v\in U$, $v\leq u$, then the $\lat$-set $\interval{v,u}=\Ucone\{v\}\cap\Lcone\{u\}$ is called
an \emph{$\lat$-interval} (or simply an \emph{interval}) in $\latU$. 

We set $[v,u]=\cut 1{\interval{v,u}}$. Thus, $[v,u]$ denotes the classical interval with respect to the $1$-cut of $\preceq$: $[v,u] = \{u'\ |\ v\leq u'\leq u\}$.

An $\lat$-set $V\in L^U$ is \emph{convex} if $V = \down V\cap\up V$. The ``$\subseteq$'' inclusion always holds as the lower set as well as upper set of $V$ always contain $V$ as a subset. For each $V\in L^U$, each of the following $\lat$-sets is convex: $\down V$, $\up V$, $\Lcone V$, $\Ucone V$. Every $\lat$-interval $\interval{v,u}$ in $\latU$ is also convex.
Every convex $\lat$-set in $\latU$ is compatible with $\approx$.

In the following two lemmas we formulate basic properties of lower and upper sets and cones that will be needed in the sequel. All the properties can be proved by direct computation.
\begin{lemma}
For each $V\in L^U$ we have
\begin{align}
\down V &= \down\down V, & \up V &= \up\up V, \label{eqn:lower_set_cosure}
\\
\Lcone V &= \down\Lcone V = \Lcone\up V,
& \Ucone V &= \up\Ucone V = \Ucone\down V. \label{eqn:lower_sets_cones_rel}
\end{align}
\end{lemma}

\begin{lemma}
For each $V,V_1,V_2\in L^U$, $u,v\in U$ we have
\begin{align}
S(V_1,V_2)&\leq S(\Lcone V_2,\Lcone V_1), & S(V_1,V_2)&\leq S(\Ucone V_2,\Ucone V_1), 
\label{eqn:cones_gal_mon}
\\
\Lcone\Ucone\Lcone V &= \Lcone V, & \Ucone\Lcone\Ucone V &= \Ucone V 
\label{eqn:cones_gal_clos} \\
V&\subseteq \Ucone\Lcone V& V&\subseteq \Lcone\Ucone V \label{eqn:cones_twice} \\
\Lcone(V_1\cup V_2) &= \Lcone V_1\cap \Lcone V_2, &
\Ucone(V_1\cup V_2) &= \Ucone V_1\cap \Ucone V_2  \label{eqn:cone_union} \\
\Lcone{\{v\}}(u)&= u\preceq v, & \Ucone \{v\}(u)&= v\preceq u, \label{eqn:single_cone_1} \\
\Lcone\Ucone\{v\} &= \Lcone\{v\}, & \Ucone\Lcone\{v\} &= \Ucone\{v\},\label{eqn:single_cone_2} \\
u\preceq v &= S(\Lcone\{u\},\Lcone\{v\}), &
u\preceq v &= S(\Ucone\{v\},\Ucone\{u\}).
\label{eqn:preceq_by_cones}
\end{align}
\end{lemma}

%

\subsection{Completely lattice $\lat$-ordered sets}
For any $\lat$-set $V\in L^U$ there exists at most one element $u\in U$ such that $\Lcone V(u) \wedge \Ucone(\Lcone V)(u) = 1$ (resp. $\Ucone V \wedge \Lcone(\Ucone V)(u) = 1$) \cite{Bel:Clofl,Bel:FRS}. If there is such an element, we call it \emph{the infimum of $V$} (resp. \emph{the supremum of $V$}) and denote $\inf V$ (resp. $\sup V$); otherwise we say that the infimum (resp.~supremum) does not exist.

If $\inf V$ exists and $V(\inf V)=1$, then it is called \emph{minimum of $V$} and denoted $\min V$. Similarly, if $\sup V$ exists and $V(\sup V)=1$, then we call it \emph{maximum of $V$} and denote $\max V$.

Infimum (supremum) of $V$ is obviously a lower (upper) bound of $V$ and, in the same time, an upper bound of $\Lcone V$ (a lower bound of $\Ucone V$).

\begin{lemma}\label{lem:inf_single_1}
If $\inf V$ exists, then $\Lcone V = \Lcone\{\inf V\}$. If $\sup V$ exists, then $\Ucone V = \Ucone\{\sup V\}$.
\end{lemma}
\begin{proof}
By definition, $\Lcone V\supseteq\{\inf V\}$. Applying both inequalities from \eqref{eqn:cones_gal_mon} we obtain $\Lcone\Ucone\Lcone V\supseteq\Lcone\Ucone\{\inf V\} $. By \eqref{eqn:cones_gal_clos}, $\Lcone\Ucone\Lcone V= \Lcone V$  and by \eqref{eqn:single_cone_2}, $\Lcone\Ucone\{\inf V\} = \Lcone\{\inf V\}$. Thus, $\Lcone V\supseteq\Lcone\{\inf V\}$.

By definition of lower cone again, $\{\inf V\}\subseteq \Ucone\Lcone V$. 
The first inequality of \eqref{eqn:cones_gal_mon} gives $\Lcone\Ucone\Lcone V\subseteq\Lcone\{\inf V\}$ and by \eqref{eqn:cones_gal_clos}, $\Lcone V\subseteq\Lcone\{\inf V\}$.

The proof for upper cones is similar.
\end{proof}

\begin{lemma}\label{lem:inf_single_2}
If $\inf V$ exists, then $V \subseteq \Ucone\{\inf V\}$. If $\sup V$ exists, then $V \subseteq \Lcone\{\sup V\}$.
\end{lemma}
\begin{proof}
By \eqref{eqn:cones_twice},  Lemma \ref{lem:inf_single_1}, \eqref{eqn:single_cone_2},
$V\subseteq \Ucone\Lcone V = \Ucone\Lcone \{\inf V\} =  \Ucone \{\inf V\}$. The second part is dual.
\end{proof}

An $\lat$-ordered set $\mathbf U$ is called \emph{completely lattice $\lat$-ordered}, if for each $V\in L^U$, both $\inf V$ and $\sup V$ exist. 

An important example of a completely lattice $\lat$-ordered set is the following. For a set $X$, the tuple $\tu{\tu{L^X,\approx^X},S}$ is a completely lattice $\lat$-ordered set with infima and suprema given by
\begin{align}
\label{eqn:power_inf}
(\inf V)(u) &= \bigwedge_{W\in L^X}V(W)\to  W(u), &
(\sup V)(u) &= \bigvee_{W\in L^X}V(W)\otimes W(u).
\end{align}
This fact follows easily e.g. from the main theorem of fuzzy concept lattices 
(fuzzy order version) \cite{Bel:FRS,Bel:Clofl}.

Note that from definition and \eqref{eqn:cones_gal_clos} it follows
\begin{align}
\inf V &= \max\Lcone V, & \sup V &= \max\Ucone V.
\end{align}
Thus, to show $\mathbf U$ is a completely lattice $\lat$-ordered set it suffices to prove existence of suprema resp.~infima of all $\lat$-sets in $U$.

Consequently, the following holds for infima and suprema of $\lat$-intervals:
\begin{align}
v &= \min \interval{v,u}, & u = \max \interval{v,u}.
\end{align}

\begin{lemma}
The following holds for any $\lat$-sets $V_1,V_2$ in a completely lattice $\lat$-ordered set $\latU$.
\begin{align}
S(\Lcone V_1,\Lcone V_2) &= \inf V_1\preceq\inf V_2, &
S(\Ucone V_1,\Ucone V_2) &= \sup V_2\preceq\sup V_1, \label{eqn:inf_cones_rel}
\\
S(V_1,V_2) &\leq \inf V_2\preceq\inf V_1, & S(V_1,V_2) &\leq \sup V_1\preceq\sup V_2
\label{eqn:inf_sets_rel}
\end{align}
\end{lemma}
\begin{proof}
By Lemma \ref{lem:inf_single_1} and \eqref{eqn:preceq_by_cones}, 
$S(\Lcone V_1,\Lcone V_2) = S(\Lcone\{\inf V_1\}, \Lcone\{\inf V_2\}) = \inf V_1\preceq\inf V_2$, proving the first part of \eqref{eqn:inf_cones_rel}. The second part is dual. \eqref{eqn:inf_sets_rel} follows from \eqref{eqn:inf_cones_rel} by \eqref{eqn:cones_gal_mon}.
\end{proof}

\begin{lemma}\label{lem:inf_two_elem}
The following holds for each $u,v\in U$:
\begin{align*}
\inf\{\deg{v\preceq u}u,v\} &= v, & \sup\{\deg{u\preceq v}u,v\}=v.
\end{align*}
\end{lemma}
\begin{proof}
By direct computation.
\end{proof}

\subsection{Isotone mappings of $\lat$-ordered sets}
We prove some basic properties of isotone mappings of $\lat$-ordered sets we will need later. The following lemma says that isotone mappings transform lower (upper) bounds of an $\lat$-set to lower (upper) bounds of its image.

\begin{lemma}
Let $f\!:U\to U'$ be an isotone mapping of $\lat$-ordered sets, $V\in L^U$. Then
\begin{align*}
f(\Lcone V)&\subseteq \Lcone f(V), & f(\Ucone V)&\subseteq \Ucone f(V).
\end{align*}
\end{lemma}
\begin{proof}
By definition of lower cone and isotony of $f$,
\begin{align*}
\Lcone V(v) &\leq V(u)\to(v\preceq u) \leq V(u)\to(f(v)\preceq f(u)),
\end{align*}
for each $u,v\in U$. Now let $u'\in U'$ and take infimum for all $u$ such that $f(u)=u'$ (in the case there is no such $u$, the infimum, as the infimum of empty set in $L$, is equal to
  $1\in L$):
\begin{align*}
\Lcone V(v) &\leq \bigwedge_{f(u)=u'}V(u)\to(f(v)\preceq f(u))
= \left(\bigvee_{f(u)=u'}V(u)\right)\to (f(v)\preceq f(u)) \\
&= f(V)(u')\to (f(v)\preceq u').
\end{align*}
Now denote $v'=f(v)$. The above inequality tells that the following holds for each $v\in U$ such that $f(v)=v'$ and $u'\in U'$:
\begin{align*}
\Lcone V(v) &\leq f(V)(u')\to (v'\preceq u').
\end{align*}
Thus,
\begin{align*}
f(\Lcone V)(v') &= \bigvee_{f(v)=v'}\Lcone V(v) \leq \bigwedge_{u'\in U'} f(V)(u')\to (v'\preceq u')
= \Lcone f(V)(v').
\end{align*}
This proves the first inclusion, the second one is similar.
\end{proof}

Let $v\in U$. As it has been said, $\Lcone V(v)$ is the degree to which $v$ is a lower bound of $V$. We have $\Lcone V(v)\leq f(\Lcone V)(f(v))$ and by the above lemma, $f(\Lcone V)(f(v))\leq \Lcone f(V)(f(v))$. This way the lemma tells that the degree to which $f(v)$ is a lower bound of $f(V)$ is greater than or equal to the degree to which $v$ is a lower bound of $V$ (and similarly for upper bounds). In the particular case  $\Lcone V(v) = 1$ (or $\Ucone V(v) = 1$) we obtain the following result:
\begin{corollary}\label{cor:isotone_image_bound}
In the setting of the previous lemma, if $\Lcone V(v)=1$, then $\Lcone f(V)(f(v))=1$ and if $\Ucone V(v)=1$, then $\Ucone f(V)(f(v))=1$. In words, if $v$ is a lower (upper) bound of $V$, then $f(v)$ is a lower (upper) bound of $f(V)$.
\end{corollary}

\begin{lemma}\label{lem:iso_cones_subset}
Let $f,g\!:U\to U$ be two mappings such that for each $u\in U$, $f(u)\leq u$ and $g(u)\geq u$. Then for each $V\in L^U$,
\begin{align}
\Lcone f(V)&\subseteq\Lcone V, & \Ucone f(V)&\supseteq\Ucone V, \\
\Lcone g(V)&\supseteq\Lcone V, & \Ucone g(V)&\subseteq\Ucone V.
\end{align}
\end{lemma}
\begin{proof}
We will prove the first inclusion only, the others being analogous. Let $v\in U$. From transitivity of $\preceq$ we have $(v\preceq f(u'))\leq (v\preceq u')$ for each $u'\in U$. Now,
\begin{align*}
\Lcone f(V)(v) &= \bigwedge_{u\in U} f(V)(u)\to (v\preceq u) = 
\bigwedge_{u\in U} \left(\bigvee_{f(u')=u}V(u')\right)\to (v\preceq u) \\
&= \bigwedge_{u\in U} \bigwedge_{f(u')=u}V(u')\to (v\preceq u)
= \bigwedge_{u'\in U}V(u')\to (v\preceq f(u')) \\
&\leq \bigwedge_{u'\in U}V(u')\to (v\preceq u') = \Lcone V(v), 
\end{align*}
proving the inclusion.
\end{proof}

In the last two lemmas we suppose $\latU$ and $\latU'$ are completely lattice $\lat$-ordered sets.

\begin{lemma}\label {lem:extensive_sup_inf}
Let $f,g$ be the same as in the previous lemma, $V\in L^U$. Then
\begin{align}
\inf f(V)&\leq \inf V, & \sup f(V)\leq \sup V, \\
\inf g(V)&\geq \inf V, & \sup g(V)\geq \sup V.
\end{align}
\end{lemma}
\begin{proof}
Follows from Lemma \ref{lem:iso_cones_subset} and \eqref{eqn:inf_cones_rel}.
\end{proof}

\begin{lemma}\label{lem:isotone_sup_inf}
Let $f\!:U\to U'$ be an isotone mapping, $V\in L^U$. Then
\begin{align}
f(\inf V)&\leq \inf f(V), & f(\sup V)&\geq \sup f(V).
\end{align}
\end{lemma}
\begin{proof}
Follows directly from Corollary \ref{cor:isotone_image_bound} (e.g., $\inf V$ is a lower bound of $V$, whence $f(\inf V)$ is a lower bound of $f(V)$ and hence is less than or equal to $\inf f(V)$).
\end{proof}

\subsection{Isotone $\lat$-Galois connections}

An \emph{isotone $\lat$-Galois connection} between $\lat$-ordered sets $\latU$ and $\latV$ \cite{GePo:Ndfc,Kon:ifgch}  is a pair $\tu{f,g}$, where $f\!:U\to V$, $g\!:V\to U$ are mappings such that for each $u\in U$, $v\in V$ it holds
\begin{align}\label{eqn:gal}
f(u)\preceq v = u\preceq g(v).
\end{align}
An isotone Galois connection between $\latU$ and $\latU$ is called simply an isotone Galois connection on $\latU$.

By isotone $\lat$-Galois connection \emph{between sets $X$ and $Y$} we understand an isotone $\lat$-Galois connection between completely lattice $\lat$-ordered sets $\lat^X$ and $\lat^Y$  \eqref{eqn:power_inf}.

Note that in \cite{GePo:Ndfc} and \cite{Kon:ifgch}, only isotone $\lat$-Galois connection between sets are considered. Thus, our approach is more general, but all results from \cite{GePo:Ndfc,Kon:ifgch} can be transferred more or less mechanically to our setting.  This is also the case in 
Section~\ref{subsec:clos_int}.

\begin{theorem}[basic properties of isotone $\lat$-Galois connections] \label{thm:bas_gal}
Let $\tu{f,g}$ be an isotone $\lat$-Galois connection between $\lat$-ordered sets $\latU$ and $\latV$. Then

(a) $u\leq g(f(u))$ for each $u\in U$, $f(g(v))\leq v$ for each $v\in V$.

(b) $f$ and $g$ are isotone.

(c) $f(g(f(u)))=f(u)$, $g(f(g(v)))=g(v)$.

(d) Let $\latU$ and $\latV$ be completely lattice $\lat$-ordered sets. For $U'\in L^U$ and $V'\in L^V$ we have 
\begin{align*}
f(\inf U')&\leq \inf f(U'), & f(\sup U')&= \sup f(U'), \\
g(\inf V')&= \inf g(V'), & g(\sup V')&\geq \sup g(V').
\end{align*}
\end{theorem}
\begin{proof}
(a)~By 
\eqref{eqn:gal}, from $f(u)\leq f(u)$ it follows $u\leq g(f(u))$ and from $g(v)\leq g(v)$ it follows $f(g(v))\leq v$.

(b)~By (a), $u_2\leq g(f(u_2))$. Thus, by transitivity, $(u_1\preceq u_2) = (u_1\preceq u_2)\otimes 1 = (u_1\preceq u_2)\otimes (u_2\preceq g(f(u_2)))  \leq (u_1\preceq g(f(u_2)))
= (f(u_1)\preceq f(u_2))$. Similarly for $g$.

(c)~By (a), $f(g(f(u)))\leq f(u)$. The opposite inequality is proved by (b) and (a): $1 = u\preceq g(f(u)) \leq f(u) \preceq f(g(f(u)))$. Similarly the second equality.

(d)~The inequalities 
$f(\inf U')\leq \inf f(U')$, $f(\sup U')\geq \sup f(U')$, $g(\inf V')\leq \inf g(V')$, $g(\sup V')\geq \sup g(V')$
follow from (b) and Lem\-ma \ref{lem:isotone_sup_inf}. 
By (a), the fourth inequality of Lemma \ref{lem:extensive_sup_inf}, and the inequality $g(\sup V')\geq \sup g(V')$ we have already proved,
\begin{align*}
\sup U' \leq \sup g(f(U')) \leq g(\sup f(U')).
\end{align*}
Now by definition, $f(\sup U') \leq \sup f(U')$. The remaining inequality, namely $g(\inf V')\geq \inf g(V')$, is proved similarly.
\end{proof}

Let $\tu{f,g}$ be
an isotone $\lat$-Galois connection between $\latU$ and $\latV$.
A pair $\tu{u,v}$, where $u\in U$ and $v\in V$, is called a \emph{fixpoint of $\tu{f,g}$} if $f(u)=v$ and $g(v)=u$. 

Suppose $\tu{u_1,v_1}$, $\tu{u_2,v_2}$ are two fixpoints of $\tu{f,g}$. We have by \eqref{eqn:gal},
\begin{align*}
u_1\preceq u_2 &= u_1\preceq g(v_2) = f(u_1)\preceq v_2 = v_1\preceq v_2
\end{align*}
and by \eqref{eqn:antisym_eq},
\begin{align*}
u_1\approx u_2 &=  v_1\approx v_2.
\end{align*}

We denote the set of all fixpoints of $\tu{f,g}$ by $\Fixfg$. For $\lat$-relations $\approx_{\Fixfg}$ and $\preceq_{\Fixfg}$ defined on $\Fixfg$ by
\begin{align}
\tu{u_1,v_1} \approx_{\Fixfg} \tu{u_2,v_2} &= u_1\approx u_2
\quad \text ( {= v_1\approx v_2}\text ), \\
\tu{u_1,v_1} \preceq_{\Fixfg} \tu{u_2,v_2} &= u_1\preceq u_2
\quad \text ( {= v_1\preceq v_2}\text ),
\end{align}
we obtain an $\lat$-ordered set $\Fixfg=\tu{\tu{\Fixfg,\approx_{\Fixfg}},\preceq_{\Fixfg}}$. In the rest of the paper, we will usually write $\approx$ instead of $\approx_{\Fixfg}$ and $\preceq$ instead of $\preceq_{\Fixfg}$.

We denote the set of all isotone Galois connections between $\lat$-ordered sets $\latU$ and $\latV$ by $\igal(\latU,\latV)$ and consider the following binary $\lat$-relations $\approx_{\igal(\latU,\latV)}$, $\preceq_{\igal(\latU,\latV)}$ on $\igal(\latU,\latV)$:
\begin{align}
\tu{f_1,g_1}\approx_{\igal(\latU,\latV)}\tu{f_2,g_2} 
&= \bigwedge_{u\in U}(f_2(u)\approx f_1(u)) \wedge \bigwedge_{v\in V}(g_1(v)\approx g_2(v)), \\
\tu{f_1,g_1}\preceq_{\igal(\latU,\latV)}\tu{f_2,g_2} 
&= \bigwedge_{u\in U}(f_2(u)\preceq f_1(u)) \wedge \bigwedge_{v\in V}(g_1(v)\preceq g_2(v)).\label{eqn:igal_preceq}
\end{align}
\begin{lemma}
$\tu{\tu{\igal(\latU,\latV),\approx_{\igal(\latU,\latV)}},\preceq_{\igal(\latU,\latV)}}$ is an $\lat$-ordered set.
\end{lemma}
\begin{proof}
Straightforward.
\end{proof}

For an $\lat$-ordered set $\latU$, an isotone $\lat$-Galois connection $\tu{f,g}$ on $\latU$ is called \emph{extensive} if
\begin{align}\label{eqn:conn_cond}
f(u)\leq u\quad\text{and}\quad g(u)\geq u
\end{align}
for each $u\in U$. The set of all extensive isotone $\lat$-Galois connections on $\latU$ is denoted $\eigal(\latU)$.

Notice that if one of the conditions \eqref{eqn:conn_cond} holds true, then
the second one follows by \eqref{eqn:gal}.

\subsection{$\lat$-closure and $\lat$-interior operators}
\label{subsec:clos_int}
Here we recall very briefly basic definitions and results on $\lat$-closure and $\lat$-interior operators. More details can be found in \cite{Bel:Lfo,Bel:Clofl,GePo:Ndfc}.

For an $\lat$-ordered set $\latU$, a mapping $C\!:U\to U$ is called an \emph{$\lat$-closure operator}, if the following holds for each $u,u_1,u_2\in U$:
\begin{align}
C(u)&\geq u, \\
C(C(u)) &= C(u), \\
u_1\preceq u_2 &\leq C(u_1)\preceq C(u_2). 
\end{align}
A mapping $I\!:U\to U$ is called an \emph{$\lat$-interior operator}, if for each $u,u_1,u_2\in U$,
\begin{align}
I(u)&\leq u, \\
I(I(u)) &= I(u), \\
u_1\preceq u_2 &\leq I(u_1)\preceq I(u_2). 
\end{align}
By $\lat$-closure (resp.~$\lat$-interior) operator \emph{on a set $X$} we mean an $\lat$-closure (resp.~$\lat$-interior) operator on the completely lattice $\lat$-ordered set 
$\lat^X$ \eqref{eqn:power_inf}.

An element $u\in U$ is a \emph{fixpoint of $C$} (resp.~\emph{fixpoint of $I$}), if $C(u)=u$ (resp.~$I(u)=u$). The set of all fixpoints of $C$ (resp.~$I$) will be denoted $\Fix_C$ (resp.~$\Fix_I$). The sets $\Fix_C$ and $\Fix_I$ inherit a structure of an $\lat$-ordered set from $\latU$. Considering $\Fix_C$ and $\Fix_I$ with this structure we have the following result:
\begin{theorem}
Let $\latU$ be a completely lattice $\lat$-ordered set. Then $\Fix_C$ is closed w.r.t.~arbitrary infima (i.e. for any $\lat$-set $V\in L^U$, $V\subseteq \Fix_C$, we have $\inf V\in \Fix_C$) and $\Fix_I$ is closed w.r.t.~arbitrary suprema (i.e. for any $\lat$-set $V\in L^U$, $V\subseteq \Fix_I$, we have $\sup V\in \Fix_I$). Consequently, $\Fix_C$ and $\Fix_I$ are completely lattice $\lat$-ordered sets.
\end{theorem}

A subset $V\subseteq U$ which is closed w.r.t.~arbitrary infima (resp.~suprema) is called an \emph{$\lat$-closure} (resp.~\emph{$\lat$-interior}) \emph{system in $\latU$}. The above theorem says that $\Fix_C$ (resp.~$\Fix_I$) is an $\lat$-closure (resp.~$\lat$-interior) system in $\latU$. In the case $\latU=\lat^X$ for some set $X$ we also talk about $\lat$-closure (resp.~$\lat$-interior) system \emph{in $X$}.

Let $\tu{f,g}$ be an isotone $\lat$-Galois connection on $\latU$. From Theorem \ref{thm:bas_gal} it easily follows that the composition $C$, given by $C(u) = g(f(u))$ is an $\lat$-closure operator on $U$ and the composition $I$,  $I(v) = f(g(v))$ is an $\lat$-interior operator on $V$.

We have the following result for the $\lat$-ordered sets of fixpoints of these operators and of the $\lat$-Galois connection $\tu{f,g}$ itself:

\begin{theorem}
Let $\latU$ be a completely lattice $\lat$-ordered set. Then the $\lat$-ordered sets $\Fixfg$, $\Fix_C$, $\Fix_I$ are isomorphic. Consequently, $\Fixfg$ is a completely lattice $\lat$-ordered set. The isomorphism $\Fixfg\to \Fix_C$ is given by $\tu{u,v}\to u$ and the isomorphism $\Fixfg\to \Fix_I$ is given by $\tu{u,v}\to v$.
\end{theorem}

\section{Power structures of $\lat$-ordered sets}\label{Sec:power}
Power structure \cite{Bri:Ps} is an algebraic structure constructed by ``lifting'' operations and relations on a (ordinary) set to its power set, i.e. the set of all its (ordinary) subsets. The theory goes back to Frobenius and recently \cite{Geo:Fps} has been generalized to a fuzzy setting.

In this section, we recall basic definitions and results from \cite{Geo:Fps} to the extent we need in this paper. We also show some results from \cite{Bel:FRS,BosMad:Ocfpr} on fuzzy power structures. Then we prove some properties of power structures of fuzzy ordered sets  we will need for the main result of this paper.

Note that in  \cite{Geo:Fps}, fuzzy power structures are studied under the framework of continuous t-norms; generalizing results we use in this paper to complete residuated lattices is straightforward. 

Let $R$ be a binary $\lat$-relation on a set $X$. We set for any $\lat$-sets $A,B\in L^X$
\begin{align}
R^\to(A,B) &= S(A,R\circ B) = \bigwedge_{x\in X}\left(A(x)\to \bigvee_{y\in X}R(x,y)\otimes B(y)\right), \\
R^\from(A,B) &= 
(R^{-1})^\to(B,A)
= S(B,R^{-1}\circ A) = S(B, A\circ R) \nonumber \\
&=  \bigwedge_{y\in X}\left(B(y)\to \bigvee_{x\in X}R(x,y)\otimes A(x)\right).
\end{align}
Since $S(A,R\circ B)$ is the degree to which $A$ is a subset of $R\circ B$, $R^\to(A,B)$ can be viewed as the degree to which each element of $A$ is related to an element of $B$.
We set
\begin{align}
R^+(A,B) &= R^\to(A,B)\wedge R^{\from}(A,B),
\end{align}
obtaining a binary $\lat$-relation, called \emph{power $\lat$-relation}, $R^+$ on the set $L^X$.
In the following, we prove some basic properties of the power $\lat$-relation $R^{+}$ for $R$ being a binary $\lat$-relation on a set $X$ and later on an $\lat$-ordered set $\tu{\tu{U,\approx},\preceq}$.

The following result is straightforward and has been proved in \cite[Theorem 4.41]{Bel:FRS}.
\begin{lemma}\label{lem:ext_basic}
For any binary $\lat$-relation   $R\in L^X$ it holds
\begin{enumerate}
\item if $R$ is reflexive, then so is $R^{+}$,
\item if $R$ is symmetric, then so is $R^{+}$,
\item if $R$ is transitive, then so is $R^{+}$.
\end{enumerate}
\end{lemma}
%
%

The following has been proved in \cite[Theorem 2]{BosMad:Ocfpr}.
\begin{theorem}\label{thm:ext_composition}
For any two $\lat$-relations $R,Q\in L^X$ it holds
\begin{align}
R^{+}\circ Q^{+} \subseteq (R\circ Q)^{+}.
\end{align}
\end{theorem}

In the next two theorems we show some basic properties of power relations of $\lat$-equivalences. We start with a lemma. Note that the $\lat$-relation $\approx^X$ on $\lat$-sets in $X$ \eqref{eq:approxx} does not depend on $\sim$.
\begin{lemma}\label{lem:equiv_compat}
Let $\sim$ be an $\lat$-equivalence on a set $X$, $A,B\in L^X$ be compatible with $\sim$. Then $A\sim^+B=A\approx^XB$.
\end{lemma}
\begin{proof}
By compatibility, $\bigvee_{x'\in X} (x\sim x')\otimes A(x')\leq A(x)$. As $(x\sim x)\otimes A(x)=A(x)$ the opposite inequality also holds and we have ${\sim}\circ A=A$. Similarly, ${\sim}\circ B=B$. Thus, 
\begin{align*}
A\sim^{+} B &= (A \sim^\to B)\wedge (A\sim^\from B)=
S(A,{\sim}\circ B)\wedge S(B,{\sim}\circ A) = S(A,B)\wedge S(B,A) \\
&= A\approx^X B.
\end{align*}
\end{proof}

\begin{theorem}\label{thm:ext_equality}
Let $\sim$ be an $\lat$-equality on a set $X$, $M\subseteq L^X$ a subset, containing only $\lat$-sets, compatible with $\sim$. Then the restriction of $\sim^{+}$ to $M$ is an $\lat$-equality on $M$ and is equal to the restriction of the $\lat$-relation $\approx^X$ \eqref{eq:approxx}.
\end{theorem}
\begin{proof}
Follows from Lemma \ref{lem:equiv_compat} as $A\approx^X B=1$ if{}f $A=B$.
\end{proof}

\begin{theorem}\label{thm:ext_compatibility} 
Let $R$ be compatible with an $\lat$-equivalence $\sim$ on $X$. Then $R^{+}$ is compatible with the power $\lat$-equivalence $\sim^{+}$. 
\end{theorem}
\begin{proof}
By Lemma \ref{lem:ext_basic}, $\sim^+$ is indeed an $\lat$-equivalence.
Compatibility of $R$ with $\sim$ means ${\sim}\circ R\circ{\sim}\subseteq R$.
By Theorem \ref{thm:ext_composition}, ${\sim^{+}}\circ R^{+}\circ{\sim^{+}}\subseteq ({\sim}\circ R\circ{\sim})^{+}\subseteq R^{+}$.
\end{proof}

The following is our main result on power relations of $\lat$-orders.
\begin{theorem}\label{thm:ext_ordered_set} 
Let $\latU = \tu{\tu{U,\approx},\preceq}$ be an $\lat$-ordered set, $M\subseteq L^U$ a subset, containing only convex $\lat$-sets in $\latU$. Then $\tu{\tu{M, \approx^{+}}, \preceq^{+}}$ is an $\lat$-ordered set.
\end{theorem}
\begin{proof}
Since convex $\lat$-sets are compatible with $\approx$, then $\approx^{+}$ is an $\lat$-equality by Theorem \ref{thm:ext_equality}. By Theorem \ref{thm:ext_compatibility}, $\preceq^{+}$ is compatible with $\approx^{+}$ and by Lemma \ref{lem:ext_basic}, $\preceq^{+}$ is reflexive and transitive.

Let $V_1,V_2\in L^U$ be convex. We have
\begin{align*}
&(V_1\preceq^+V_2)\wedge(V_2\preceq^+V_1) = 
(V_1 \preceq^\to V_2)\wedge (V_1\preceq^\from V_2)\wedge
(V_2\preceq^\to V_1)\wedge (V_2\preceq^\from V_1)
\\
&\quad = 
(V_1 \preceq^\to V_2)\wedge
(V_2\preceq^\to V_1)\wedge (V_2\preceq^\from V_1)\wedge (V_1\preceq^\from V_2)
\\
&\quad = S(V_1,\down V_2)\wedge S(V_1,\up V_2) \wedge
S(V_2,\down V_1)\wedge S(V_2,\up V_1)
\\
&\quad = S(V_1,\down V_2\cap\up V_2) \wedge S(V_2,\down V_1\cap\up V_1)
= S(V_1,V_2)\wedge S(V_2,V_1) = V_1\approx^X V_2
\\ 
&\quad = V_1\approx^+ V_2
\end{align*}
(the last equality follows by Lemma \ref{lem:equiv_compat}), proving antisymmetry.
\end{proof}

The following two lemmas show a way of efficient computing values of power relations $\approx^+$ and $\preceq^+$ on intervals.
\begin{lemma}\label{lem:ext_preceq} 
Let $\latU = \tu{\tu{U,\approx},\preceq}$ be an $\lat$-ordered set, $V_1,V_2\in L^U$ two $\lat$-sets having minimum and maximum, $\min V_1=u_1$, $\max V_1 = v_1$, $\min V_2=u_2$, $\max V_2 = v_2$. Then $V_1 \preceq^{+} V_2 = (u_1\preceq u_2)\wedge(v_1\preceq v_2)$.
\end{lemma}
\begin{proof}
We will first prove that for each $w\in U$,
\begin{align}
({\preceq}\circ V_2)(w) &= w\preceq v_2.\label{eqn:ext_preceq_minmax}
\end{align}
Since $v_2$ is the maximum of $V_2$, then by \eqref{eqn:cones_twice}, Lemma \ref{lem:inf_single_2}, and \eqref{eqn:single_cone_2},  $V_2\subseteq \Lcone\Ucone V_2 = \Lcone\Ucone\{v_2\}= \Lcone\{v_2\}$ whence $V_2(w')\leq \Lcone\{v_2\}(w') = w'\preceq v_2$ \eqref{eqn:single_cone_1}. Thus, for each $w\in U$, 
\begin{align*}
(w\preceq w')\otimes V_2(w')\leq (w\preceq w')\otimes(w'\preceq v_2)\leq(w\preceq v_2)
\end{align*}
by transitivity. Taking supremum through all $w'$ on the left-hand side and taking into account that $(w\preceq v_2)\otimes V_2(v_2) = w\preceq v_2$ we obtain \eqref{eqn:ext_preceq_minmax}.

Thus, 
\begin{align*}
(V_1\preceq^\to V_2) = \bigwedge_{w\in U}V_1(w)\to (w\preceq v_2)
= \Ucone V_1(v_2) = \Ucone\{v_1\}(v_2) = v_1\preceq v_2
\end{align*}
by \eqref{eqn:single_cone_1}.

One can prove similarly $(V_1\preceq^\from V_2) = u_1\preceq u_2$ and obtain the desired equality.
\end{proof}

\begin{lemma}\label{lem:ext_approx} 
Let $\latU = \tu{\tu{U,\approx},\preceq}$ be an $\lat$-ordered set, $V_1=\interval{u_1,v_1}$, $V_2=\interval{u_2,v_2}$ intervals in $U$. Then $V_1 \approx^{+} V_2 = (u_1\approx u_2)\wedge(v_1\approx v_2)$.
\end{lemma}
\begin{proof}
According to Theorem \ref{thm:ext_ordered_set}, $\tu{\tu{M,\approx^+},\preceq^+}$, where $M=\{V_1,V_2\}$, is an $\lat$-ordered set. Thus by \eqref{eqn:antisym_eq} and Lemma \ref{lem:ext_preceq},
\begin{align*}
V_1 \approx^{+} V_2 &= (V_1\preceq^+ V_2)\wedge(V_2\preceq^+ V_1)\\
&= (u_1\preceq u_2)\wedge(v_1\preceq v_2)\wedge(u_2\preceq u_1)\wedge(v_2\preceq v_1) \\
&=  (u_1\approx u_2)\wedge(v_1\approx v_2).
\end{align*}
%
\end{proof}

\section{Complete $\lat$-relations}\label{Sec:compl}
In classical theory of complete lattices (see for example \cite{GanWi:FCA}), a binary relation $R$ on a complete lattice $\latU$ is called complete, if for each system $\{\tu{u_j,v_j}\}_{j\in J}$ of pairs of elements of $U$ from $u_j\, R\, v_j$ for each $j\in J$ it follows $\left(\bigwedge_{j\in J} u_j\right)\, R\, \left(\bigwedge_{j\in J} v_j\right)$ and $\left(\bigvee_{j\in J} u_j\right)\, R\, \left(\bigvee_{j\in J} v_j\right)$.

It can be easily checked that the following condition is equivalent to the above condition of completeness of $R$: if $V_1,V_2\subseteq U$ are such that for each $v_1\in V_1$ there is $v_2\in V_2$ such that $v_1\, R\, v_2$ and for each $v_2\in V_2$ there is $v_1\in V_1$ such that $v_1\, R\, v_2$, then $\left(\bigwedge V_1\right)\, R\, \left(\bigwedge V_2\right)$ and $\left(\bigvee V_1\right)\, R\, \left(\bigvee V_2\right)$.


This leads us to the following definition. A binary $\lat$-relation on a completely lattice $\lat$-ordered set $\latU=\tu{\tu{U,\approx},\preceq}$ is called \emph{complete}, if it is compatible with $\approx$ and for any two $\lat$-sets $V_1,V_2\in L^U$ it holds
\begin{align}
R^+(V_1,V_2) &\leq R(\inf V_1,\inf V_2), \\
R^+(V_1,V_2) &\leq R(\sup V_1,\sup V_2).
\end{align}

Following are basic properties of complete relations on a completely lattice $\lat$-ordered set $\latU=\tu{\tu{U,\approx},\preceq}$.

\begin{lemma}\label{lem:inv_compl}
If $R$ is complete, then so is $R^{-1}$.
\end{lemma}
\begin{proof}
We have
\begin{align*}
(R^{-1})^+(V_1,V_2) &= (R^{-1})^\to(V_1,V_2)\wedge (R^{-1})^\from(V_1,V_2) \\
&= R^\from(V_2,V_1)\wedge R^\to(V_2,V_1) 
= R^+(V_2,V_1)\leq R(\inf V_2,\inf V_1) \\ &= R^{-1}(\inf V_1,\inf V_2),
\end{align*}
and similarly for suprema.
\end{proof}
\begin{theorem}\label{thm:compl_closure_system}
The system of all complete binary $\lat$-relations on $\latU$ is an $\lat$-closure system
in the set $U\times U$, hence a completely lattice $\lat$-ordered set.
\end{theorem}
\begin{proof}
We will show 1. that if $R_j$, $j\in J$, are complete then so is $\bigcap_{j\in J} R_j$ and 2. that for each $a\in L$ and $R$ complete the shift $a\to R$ is also complete. Since the system of all binary $\lat$-relations that are compatible with $\approx$ is an $\lat$-closure system, there is no need to prove compatibility of the relations.

1.~We have 
\begin{align*}
\left(\bigcap_j R_j\right)\circ V(v) &= \bigvee_{w\in U}\left(\bigwedge_j R_j(v,w)\right)\otimes V(w)
\leq \bigwedge_j\bigvee_{w\in U}R_j(v,w)\otimes V(w) \\
&= \bigwedge_j(R_j\circ V)(v).
\end{align*}
Thus, $(\bigcap_j R_j)\circ V\subseteq \bigcap_j(R_j\circ V)$. Now, 
\begin{align*}
\left(\bigcap_j R_j\right)^\to(V_1,V_2) &= S\left(V_1, \left(\bigcap_j R_j\right)\circ V_2\right)\leq 
S\left(V_1,\bigcap_j\left(R_j\circ V_2\right)\right) = \bigwedge_jS(V_1,R_j\circ V_2)\\
&= \bigwedge_j(R_j)^\to(V_1, V_2)
\end{align*}
and, finally, 
\begin{align*}
&\left(\bigcap_jR_j\right)^+(V_1,V_2) = \left(\bigcap_jR_j\right)^\to(V_1,V_2) \wedge 
 \left(\bigcap_jR_j\right)^\from(V_1,V_2) \\
&\quad 
\leq \bigwedge_j (R_j)^\to(V_1,V_2)\wedge 
(R_j)^\from(V_1,V_2) =
 \bigwedge_j(R_j)^+(V_1,V_2) \\
&\quad \leq \bigwedge_j R_j(\inf V_1,\inf V_2) = \left(\bigcap_j R_j\right)(\inf V_1,\inf V_2).
\end{align*}
Similarly for suprema.
 
 2.~We have 
 \begin{align*}
((a\to R)\circ V)(v) &= \bigvee_{w\in U}(a\to R(v,w))\otimes V(w)
 \leq a\to \bigvee_{w\in U}R(v,w)\otimes V(w) \\
 &= a\to (R\circ V)(v).
\end{align*}
Thus, $(a\to R)\circ V\subseteq a\to (R\circ V)$. Now, 
\begin{align*}
(a\to R)^\to(V_1,V_2)=S(V_1,(a\to R)\circ V_2) \leq 
 S(V_1,a\to (R\circ V_2)) =
 a\to S(V_1,R\circ V_2)
\end{align*}
and, finally, 
\begin{align*}
&(a\to R)^+(V_1,V_2) = 
 (a\to R)^\to(V_1,V_2)\wedge (a\to R)^{\from}(V_1,V_2) \\
 &\quad \leq
(a\to R^\to(V_1,V_2))\wedge(a\to (R^{\from}(V_1,V_2)) =
a\to (R^\to(V_1,V_2)\wedge R^{\from}(V_1,V_2)) \\
&\quad \leq
a\to R(\inf V_1,\inf V_2) = (a\to R)(\inf V_1,\inf V_2).
\end{align*}
Similarly for suprema.
\end{proof}

\begin{lemma}\label{lem:preceq_approx_compl}
The following holds for each $V_1,V_2\in L^U$:
\begin{align*}
V_1 \preceq^\to V_2 &\leq  \sup V_1\preceq\sup V_2, &
 V_1\preceq^\from V_2&\leq \inf V_1\preceq \inf V_2.
\end{align*}
\end{lemma}
\begin{proof}
We have by \eqref{eqn:cones_gal_mon}, \eqref{eqn:lower_sets_cones_rel}, \eqref{eqn:inf_cones_rel},
\begin{align*}
V_1\preceq^\to V_2 &= S(V_1, {\preceq}\circ V_2) = S(V_1,\down V_2) \leq
S(\Ucone\down V_2,\Ucone V_1) = S(\Ucone V_2,\Ucone V_1) \\
&=
 \sup V_1\preceq\sup V_2.
\end{align*}
Hence the first inequality.
The second one is obtained similarly.
\end{proof}

\begin{theorem}\label{thm:preceq_approx_compl}
The $\lat$-relations $\preceq$ and $\approx$ on $\latU$ are complete.
\end{theorem}
\begin{proof}
By Lemma \ref{lem:preceq_approx_compl}, for each $V_1,V_2\in L^U$, $V_1\preceq^+ V_2 \leq (V_1\preceq^\to V_2) \leq  \sup V_1\preceq\sup V_2$ and $V_1\preceq^+ V_2 \leq  V_1\preceq^\from V_1\leq \inf V_1\preceq \inf V_2$, proving completeness of $\preceq$.

Since $\approx = {\preceq}\cap{\succeq}$, completeness of $\approx$ follows from Lemma \ref{lem:inv_compl} and Theorem \ref{thm:compl_closure_system}.
\end{proof}

\section{Complete tolerances}\label{Sec:tol}
\subsection{Basic properties}
Recall that $\lat$-tolerance on a set $X$ is a reflexive and symmetric binary $\lat$-relation on $X$.
For an $\lat$-tolerance $\sim$ on a set $X$, an $\lat$-set $B\in L^X$ is called a \emph{block of  $\sim$} \cite{Bel:FRS} if for each $x_1,x_2\in X$ it holds
$B(x_1)\otimes B(x_2)\leq (x_1\sim x_2)$.
A block $B$ is called \emph{maximal} if for each block $B'$ from $B\subseteq B'$ it follows $B=B'$. The set of all maximal blocks of $\sim$ 
always exists by Zorn's lemma, 
is called \emph{the factor set of $X$ by $\sim$} and denoted by $X/{\sim}$.

Further we set for each $x\in X$,
$\class x (y)= x\sim y$,
obtaining an $\lat$-set $\class x$ called \emph{the class of $\sim$ determined by $x$}.

Let $\sim$ be a complete tolerance on a completely lattice $\lat$-ordered set $\latU=\tu{\tu{U,\approx},\preceq}$. From reflexivity of $\sim$ we have $V\subseteq {\sim}\circ V$ for each $V\in L^U$ and from symmetry ${\sim}^{-1} = {\sim}$.

For each $u\in U$ we set
\begin{align}
u_\sim &= \inf\class u, & u^\sim &= \sup\class u.
\end{align}

We denote the system of all complete $\lat$-tolerances on a completely lattice $\lat$-ordered set $\latU$ by $\ctol\latU$ and consider it together with the $\lat$-equality $\approx^{U\times U}$ and $\lat$-order $S$.
\begin{theorem}\label{thm:ctol_clos}
$\ctol\latU$ is an $\lat$-closure system
in the set $U\times U$, hence a completely lattice $\lat$-ordered set.
\end{theorem}
\begin{proof}
Evidently, if $\sim$ is an $\lat$-tolerance then so is $a\to{\sim}$ for each $a\in L$ and if $\sim_j$, $j\in J$, are $\lat$-tolerances then $\bigcap_{j\in J}$ is also an $\lat$-tolerance. Thus, the theorem follows from Theorem \ref{thm:compl_closure_system}.
\end{proof}

\subsection{From complete tolerances to isotone Galois connections}

\begin{lemma}\label{lem:downsim}
For each $u\in U$, $u\sim u_\sim = u\sim u^\sim = 1$.
\end{lemma}
\begin{proof}
Set $V_1=\{u\}$, $V_2 = \class u$. Since $V_1\subseteq V_2$,
we have
  $V_1\sim^\to V_2=1$. Further, $({\sim}\circ V_1)(v) = v\sim u = \class u(v)$. Thus, $V_1\sim^\from V_2 = S(\class u,\class u) = 1$. Now,
\begin{align*}
V_1\sim^+ V_2 &= (V_1\sim^\to V_2)\wedge (V_1\sim^\from V_2)=1
\end{align*}
and by completeness of $\sim$, $1=\inf V_1\sim \inf V_2= u\sim u_\sim$ and $1=\sup V_1\sim \sup V_2= u\sim u^\sim$.
\end{proof}

\begin{lemma}\label{lem:double_sim}
It holds for any $u\in U$
\begin{align}
{u_\sim}^\sim &\geq u, & {u^\sim}_\sim &\leq u.
\end{align}
\end{lemma}
\begin{proof}
By Lemma \ref{lem:downsim}, $\class u(u_\sim)=1$. This means that also $\class {u_\sim}(u)=1$. Since ${u_\sim}^\sim = \sup \class{u_\sim}$, we have the first inequality.

The second inequality is analogous.
\end{proof}

\begin{lemma}\label{lem:two_downsim}
For each $u,v\in U$ it holds
\begin{align}
(u\preceq v) &\leq (u_\sim\preceq v_\sim), & (u\preceq v) &\leq (u^\sim\preceq v^\sim).
\label{eqn:two_downsim}
\end{align}
\end{lemma}
\begin{proof}
Let $a=u\preceq v$, $V_1=\{\deg au,v\}$, $V_2=\{\deg a{u^\sim},v^\sim\}$. By Lemma \ref{lem:downsim}, $u\sim u^\sim = v\sim v^\sim = 1$. Thus, $V_1\subseteq {\sim}\circ V_2$, $V_2\subseteq {\sim}\circ V_1$ and we have $V_1\sim^+ V_2=1$. By completeness of $\sim$, 
$\sup V_1\sim \sup V_2 = 1.\label{eqn:lem:two_downsim1}$

By Lemma \ref{lem:inf_two_elem}, $\sup V_1=v$. Thus,
 $v\sim\sup V_2=1$, which means $\sup V_2\leq v^\sim$. On the other hand, since $V_2(v^\sim)=1$, we have $\sup V_2\geq v^\sim$, whence $\sup V_2= v^\sim$. By Lemma \ref{lem:inf_single_2}, $V_2\subseteq \Lcone\{v^\sim\}$. Thus, $a=V_2(u^\sim)\leq \Lcone\{v^\sim\}(u^\sim) = u^\sim\preceq v^\sim$ and the second inequality in \eqref{eqn:two_downsim} is proved.

The first inequality is proved similarly.
\end{proof}

\begin{theorem}\label{thm:from_tol_to_gal}
The pair $\tu{{}_\sim,{}^\sim}$ is an extensive isotone Galois connection on $\latU$.
\end{theorem}
\begin{proof}
We will show $\tu{{}_\sim,{}^\sim}$ is an isotone Galois connection. Let $u,v\in U$. We have by Lemma \ref{lem:two_downsim}, Lemma \ref{lem:double_sim}, and transitivity of $\preceq$,
\begin{align*}
(u_\sim\preceq v)&\leq ({u_\sim}^\sim\preceq v^\sim)\leq (u\preceq v^\sim).
\end{align*}
The converse inequality is proven analogously.

Extensivity of $\tu{{}_\sim,{}^\sim}$ follows trivially from reflexivity of $\preceq$.
\end{proof}

\subsection{Structure of maximal blocks}

\begin{lemma}\label{lem:int_blocks}
If $\tu{u,v}$ is a fixpoint of $\tu{{}_\sim,{}^\sim}$, then $\interval{v,u}$ is a block of $\sim$.
\end{lemma}
\begin{proof}
1.~We will prove that for each $w$,
\begin{align}\label{eqn:int_blocks}
\interval{v,u}(w) \leq u\sim w
\end{align}
(i.e. ``if $w$ belongs to $\interval{v,u}$, then it is similar to $u$'').

Set $a=w\preceq u$, $b=v\preceq w$, $V_1=\{u,\deg aw\}$, $V_2=\{\deg bv,w\}$. 
By Lemma \ref{lem:inf_two_elem}, $\sup V_1=u$ and $\sup V_2 = w$. 

Now, 
\begin{align*}
({\sim}\circ V_1)(v) &= ((v\sim u)\otimes V_1(u))\vee ((v\sim w)\otimes V_1(w)) 
= (1\otimes 1)\vee ((v\sim w)\otimes a)  = 1, \\
({\sim}\circ V_1)(w) &= ((w\sim u)\otimes V_1(u))\vee ((w\sim w)\otimes V_1(w))
= ((w\sim u)\otimes 1)\vee (1\otimes a) \geq a, \\
({\sim}\circ V_2)(u) &= ((u\sim v)\otimes V_2(v))\vee ((u\sim w)\otimes V_2(w))
= (1\otimes b)\vee((u\sim w)\otimes 1)\geq b, \\
({\sim}\circ V_2)(w) &= ((w\sim v)\otimes V_2(v))\vee ((w\sim w)\otimes V_2(w))
= ((w\sim v)\otimes b) \vee (1\otimes 1) = 1.
\end{align*}
Thus, 
\begin{align*}
S(V_1,{\sim}\circ V_2) &= 
(V_1(u)\to ({\sim}\circ V_2)(u))\wedge (V_1(w)\to ({\sim}\circ V_2)(w)) \geq b,\\
S(V_2,{\sim}\circ V_1) &= 
(V_2(v)\to ({\sim}\circ V_1)(v))\wedge (V_2(w)\to ({\sim}\circ V_1)(w)) \geq a
\end{align*}
and
\begin{align*}
\interval{v,u}(w) &=a\wedge b\leq S(V_1,{\sim}\circ V_2)\wedge S(V_2,{\sim}\circ V_1)= V_1\sim^+ V_2 \\ 
&\leq \sup V_1\sim \sup V_2 = u\sim w,
\end{align*}
proving 
\eqref{eqn:int_blocks}.

2.~Let $w_1,w_2\in U$, $a_1=\interval{v,u}(w_1)$, $a_2=\interval{v,u}(w_2)$, $b_1 = w_1\preceq u$, $b_2 = w_2\preceq u$. By \eqref{eqn:int_blocks}, $a_1\leq b_1$, $a_2\leq b_2$.
Set $V_1=\{\deg {b_1}u,w_1\}$, $V_2 = \{\deg {b_2}u,w_2\}$. By similar direct calculations as above we obtain
\begin{align*}
a_1\otimes a_2 \leq a_1\otimes b_2 = (b_1\to b_2)\wedge(1\to b_2\otimes a_1)
\leq V_1\sim^\to V_2.
\end{align*}
Similarly, $a_1\otimes a_2\leq V_1\sim^\from V_2$ and
\begin{align*}
a_1\otimes a_2\leq V_1\sim^+ V_2 \leq \inf V_1\sim \inf V_2 = w_1\sim w_2,
\end{align*}
proving $\interval{v,u}$ is a block.
\end{proof}

\begin{lemma}\label{lem:block_inf}
If $B$ is a block of $\sim$, then so is $B\cup \{\inf B\}$.
\end{lemma}
\begin{proof}
Let $u=\inf B$. It suffices to prove $B(v)\leq u\sim v$ for each $v\in U$. 

Let $V = \{v\}$. We have
\begin{align*}
B\sim^\to V &= \bigwedge_{w\in U}B(w)\to(v\sim w) \geq B(v), \\
B\sim^\from V &= B(v).
\end{align*}
Thus,
\begin{align*}
B(v) &\leq B\sim^+ V \leq \inf B\sim \inf V = u\sim v
\end{align*}
and the lemma is proved.
\end{proof}

\begin{lemma}
For each block $B$ of $\sim$ there is a fixpoint $\tu{u,v}$ of $\tu{{}_\sim,{}^\sim}$ such that $B\subseteq\interval{v,u}$.
\end{lemma}
\begin{proof}
Set $w=\inf B$, $u=w^\sim$, $v=u_\sim$. Since $\tu{{}_\sim,{}^\sim}$ is an isotone Galois connection (Theorem \ref{thm:from_tol_to_gal}), $\tu{u,v}$ is a fixpoint. By Lemma \ref{lem:block_inf}, the $\lat$-set $B'=B\cup\{w\}$ is again a block. By Lemma \ref{lem:inf_single_2}, we have $B'\subseteq \Ucone\{w\}\subseteq\Ucone\{v\}$. Also, for each $w'$ it holds $B'(w')\leq w'\sim w$. Thus, by definition of class, $B'\subseteq\class w$, whence $\sup B'\leq\sup\class w=u$ \eqref{eqn:inf_sets_rel}. This yields $B'\subseteq\Lcone\{u\}$ and we can conclude $B\subseteq B'\subseteq \Ucone\{v\}\cap\Lcone\{u\}=\interval{v,u}$.
\end{proof}

\begin{theorem}\label{thm:max_blocks}
Maximal blocks of $\sim$ are exactly intervals $\interval{v,u}$, where $\tu{u,v}$ are fixpoints of $\tu{{}_\sim,{}^\sim}$.
\end{theorem}
\begin{proof}
Follows from the above lemmas.
\end{proof}

\subsection{Structure of classes}
\begin{theorem}\label{thm:classes}
For each $u\in U$, the class $\class u$ is equal to the interval $\interval{u_\sim,u^\sim}$.
\end{theorem}
\begin{proof}
By Lemma \ref{lem:inf_single_2}, $\class u\subseteq \Ucone \{\inf \class u\} = \Ucone\{u_\sim\}$ and similarly $\class u\subseteq \Lcone\{u^\sim\}$. Thus, $\class u\subseteq \interval{u_\sim,u^\sim}$.

Let $u'\in U$, $a =\interval{u_\sim,u^\sim}(u') = (u_\sim\preceq u')\wedge(u'\preceq u^\sim)$. We will show the $\lat$-set $V=\{\deg a{u'},u\}$ is a block. For the lower cone of $V$ we have
\begin{align}\label{eqn:thm:classes}
\Lcone V(w) &= (w\preceq u)\wedge(a\to (w\preceq u')).
\end{align}
Let $v=\inf V$. According to Lemma \ref{lem:inf_single_2}, $V\subseteq \Ucone \{v\}$. By Lemma \ref{lem:double_sim}, ${v^\sim}_\sim\leq v$, whence $V\subseteq \Ucone \{{v^\sim}_\sim\}$.

Now consider membership degrees of $u$ and $u'$ in the lower cone $\Lcone\{v^\sim\}$. Since $a\leq u_\sim \preceq u'$, then \eqref{eqn:thm:classes} $\Lcone V(u_\sim) = 1\wedge (a\to (u_\sim\preceq u')) = 1.$ Thus, $1= u_\sim\preceq v = u\preceq v^\sim =\Lcone\{v^\sim\}(u)$, obtaining $\Lcone\{v^\sim\}(u)=1$.

For $\Lcone\{v^\sim\}(u')$ we first notice by \eqref{eqn:thm:classes}, $\Lcone V(u'_\sim) = u'_\sim\preceq u = u'\preceq u^\sim\geq a$. By Lemma \ref{lem:inf_single_2}, and \eqref{eqn:single_cone_1}, $\Lcone V(u'_\sim) = \Lcone \{v\}(u'_\sim) = u'_\sim\preceq v$, and $\Lcone\{v^\sim\}(u') = u'\preceq v^\sim = u'_\sim \preceq v\geq a$. Thus, $V\subseteq \Lcone\{v^\sim\}$.

Together, $V\subseteq \Ucone \{{v^\sim}_\sim\} \cap \Lcone\{v^\sim\} = \interval{{v^\sim}_\sim,v^\sim}$. By Theorem \ref{thm:max_blocks}, $\interval{{v^\sim}_\sim,v^\sim}$ is a block. Thus, $V$ is also a block and by definition of block we obtain $\interval{u_\sim,u^\sim}(u') = V(u') = V(u')\otimes V(u)\leq u'\sim u = \class u(u')$. Thus, $\interval{u_\sim,u^\sim}\subseteq \class u$ and the theorem is proved.
\end{proof}

The following is an important consequence of Theorem \ref{thm:classes} that we will use later to prove our main result.

\begin{lemma}\label{lem:tol_gal_tol}
For each $u,v\in U$ we have
\begin{align}
u\sim v &= (u_\sim \preceq v)\wedge(v\preceq u^\sim).
\end{align}
\end{lemma}
\begin{proof}
The right-hand side is equal to $\interval{u_\sim,u^\sim}(v)$, which is by Theorem \ref{thm:classes} equal to $\class u(v) = u\sim v$.
\end{proof}

We use the above results in the proof of the following lemma. By Theorem \ref{thm:from_tol_to_gal}, for each complete $\lat$-tolerance $\sim$ on $\latU$ the pair $\tu{{}_\sim,{}^\sim}$ is an isotone $\lat$-Galois connection. Thus, we can $\lat$-order such $\lat$-Galois connections by the $\lat$-relation  $\preceq_{\igal(\latU,\latU)}$ \eqref{eqn:igal_preceq}.
\begin{lemma}\label{lem:tol_to_gal_iso}
For any two complete $\lat$-tolerances $\sim_1$, $\sim_2$ on $\latU$ we have
\begin{align*}
S(\sim_1,\sim_2) &= \tu{{}_{\sim_1},{}^{\sim_1}}\preceq_{\igal(\latU,\latU)}\tu{{}_{\sim_2},{}^{\sim_2}}.
\end{align*}
\end{lemma}
\begin{proof}
By definitions of $S$ and $\preceq_{\igal(\latU,\latU)}$ we have to prove the following equality:
\begin{align}\label{eqn:tol_to_gal_iso_aux}
\bigwedge_{u,v\in U}(u\sim_1 v)\to (u\sim_2 v)
&=
\bigwedge_{u\in U}(u_{\sim_2}\preceq u_{\sim_1}) \wedge \bigwedge_{u\in U}(u^{\sim_1}\preceq u^{\sim_2}).
\end{align}
We will proceed by proving both inequalities ``$\leq$'' and ``$\geq$''.

``$\leq$'':~Since $u\sim_1 u^{\sim_1}$ (Lemma \ref{lem:downsim}), the left-hand side of 
\eqref{eqn:tol_to_gal_iso_aux} is $\leq$ $\bigwedge_{u\in U}u\sim_2 u^{\sim_1}$. Now by Theorem \ref{thm:classes} and \eqref{eqn:single_cone_1} we have
\begin{align*}
u\sim_2 u^{\sim_1} &= \class[\sim_2]{u^{\sim_1}}(u) =
\Lcone\{(u^{\sim_1})^{\sim_2}\}(u)\wedge\Ucone\{(u^{\sim_1})_{\sim_2}\}(u)
\\
&\leq \Ucone\{(u^{\sim_1})_{\sim_2}\}(u) = (u^{\sim_1})_{\sim_2}\preceq u =
u^{\sim_1}\preceq u^{\sim_2}.
\end{align*}
Thus, $\bigwedge_{u,v\in U}(u\sim_1 v)\to (u\sim_2 v)\leq \bigwedge_{u\in U}u^{\sim_1}\preceq u^{\sim_2}$. The inequality $\bigwedge_{u,v\in U}(u\sim_1 v)\to (u\sim_2 v)\leq \bigwedge_{u\in U}u_{\sim_2}\preceq u_{\sim_1}$ is proved similarly.

``$\geq$'': by Theorem \ref{thm:classes} and \eqref{eqn:single_cone_1} again and by antisymmetry of $\preceq$ we obtain
\begin{align*}
&(u^{\sim_1}\preceq u^{\sim_2})\otimes(u\sim_1 v)
\leq (u^{\sim_1}\preceq u^{\sim_2})\otimes
\left((v\preceq u^{\sim_1})\wedge (u_{\sim_1}\preceq v)\right) \\
&\quad \leq (u^{\sim_1}\preceq u^{\sim_2})\otimes
(v\preceq u^{\sim_1})\leq v\preceq u^{\sim_2}.
\end{align*}
Similarly $(u_{\sim_2}\preceq u_{\sim_1})\otimes (u\sim_1 v)\leq (u_{\sim_2}\preceq v)$, thereby (Theorem \ref{thm:classes} and \eqref{eqn:single_cone_1})
\begin{align*}
(u\sim_2 v) &= (u_{\sim_2}\preceq v)\wedge (v\preceq u^{\sim_2})\\
&\geq \left( (u^{\sim_1}\preceq u^{\sim_2})\otimes(u\sim_1 v)\right)
\wedge
\left( (u_{\sim_2}\preceq u_{\sim_1})\otimes (u\sim_1 v)\right) \\
&\geq \left( (u^{\sim_1}\preceq u^{\sim_2})\wedge (u_{\sim_2}\preceq u_{\sim_1})\right)
\otimes (u\sim_1 v).
\end{align*}
By adjointness,
\begin{align*}
(u\sim_1 v)\to (u\sim_2 v)\geq (u^{\sim_1}\preceq u^{\sim_2})\wedge (u_{\sim_2}\preceq u_{\sim_1}),
\end{align*}
yielding the ``$\geq$'' part of \eqref{eqn:tol_to_gal_iso_aux}.
\end{proof}

\subsection{From extensive isotone Galois connections to complete tolerances}

Let $\tu{f,g}$ be an extensive isotone $\lat$-Galois connection on a completely lattice $\lat$-ordered set $\latU = \tu{\tu{U,\approx},\preceq}$. We set for each $u,v\in U$,
\begin{align}
u\simfg v &= (f(u)\preceq v)\wedge(v\preceq g(u)).
\end{align}
The following theorem summarizes main properties of the $\lat$-relation $\simfg$.

\begin{theorem}\label{thm:from_gal_to_tol}
$\simfg$ is a complete tolerance such that for each $u\in U$,
\begin{align}
u_{\simfg} &= f(u), & u^{\simfg} &= g(u).\label{eqn:simfg_complete}
\end{align}
\end{theorem}
\begin{proof}
The $\lat$-relation $\simfg$ is evidently reflexive and symmetric, hence an $\lat$-tolerance.

Set $R(u,v)=u\preceq g(v)$. We have $u\simfg v = R(u,v)\wedge R^{-1}(u,v)$. Thus, by
Lemma~\ref{lem:inv_compl} and Theorem~\ref{thm:compl_closure_system} it is sufficient to prove that $R$ is complete.

Let $V\in L^U$.  Using obvious inequality $V(w)\leq g(V)(g(w))$ we have
\begin{align*}
(R\circ V)(v) &= \bigvee_{w\in U}R(v,w)\otimes V(w) = 
\bigvee_{w\in U}(v\preceq g(w))\otimes V(w) \\ 
&\leq
\bigvee_{w\in U}(v\preceq g(w))\otimes g(V)(g(w)) \leq
\bigvee_{w'\in U}(v\preceq w')\otimes g(V)(w') \\
&= ({\preceq}\circ g(V))(v)
\end{align*} 
and
\begin{align*}
(R^{-1}\circ V)(v) &= \bigvee_{w\in U}R(w,v)\otimes V(w) = 
\bigvee_{w\in U}(v\succeq f(w))\otimes V(w) \\ 
&\leq
\bigvee_{w\in U}(v\succeq f(w))\otimes f(V)(f(w)) \leq
\bigvee_{w'\in U}(v\succeq f(w))\otimes f(V)(w') \\
&= ({\succeq}\circ f(V))(v),
\end{align*} 
whence $R^\to(V_1,V_2) = S(V_1,R\circ V_2)\leq S(V_1,{\preceq}\circ g(V_2)) =
V_1\preceq^\to g(V_2)$ and $R^{\from}(V_1,V_2) = S(V_2,R^{-1}\circ V_1)\leq S(V_2,{\succeq}\circ f(V_1)) =
f(V_1)\preceq^\from V_2$.

Now by Lemma \ref{lem:preceq_approx_compl} and Theorem \ref{thm:bas_gal} (d),
\begin{align*}
R^+(V_1,V_2) &\leq R^\to(V_1,V_2) \leq V_1\preceq^\to g(V_2)
\leq \sup V_1\preceq \sup g(V_2) \\
&\leq \sup V_1\preceq g(\sup V_2)
= R(\sup V_1,\sup V_2),\\
R^+(V_1,V_2) &\leq R^\from(V_1,V_2) \leq  f(V_1)\preceq^\from V_2
\leq  \inf f(V_1)\preceq \inf V_2 \\
&\leq f(\inf V_1)\preceq\inf V_2
= R(\inf V_1,\inf V_2),
\end{align*}
proving completeness of $R$ and hence of $\simfg$.

To prove \eqref{eqn:simfg_complete}, we notice that for each $u\in U$ the class $\class[{\simfg}] u$ is equal to the interval $\interval{f(u),g(u)}$:
\begin{align*}
\class[{\simfg}] u(v) &= u\simfg v = (f(u)\preceq v)\wedge(v\preceq g(u)) \\
&= \Ucone\{f(u)\}(v)\wedge \Lcone\{g(u)\}(v) = \interval{f(u),g(u)}(v).
\end{align*}
Now, $u_{\simfg} = \inf \interval{f(u),g(u)} = f(u)$ and $u^{\simfg} = \sup \interval{f(u),g(u)} = g(u)$.

\end{proof}

\subsection{Factorization theorem, representation theorem}
By Theorem \ref{thm:max_blocks}, the factor set $U/{\sim}$ consists of intervals. Thus, by Theorem \ref{thm:ext_ordered_set}, the tuple $\latU/{\sim} = \tu{\tu{U/{\sim},\approx^{+}},\preceq^{+}}$ is an $\lat$-ordered set. 
By Theorem \ref{thm:from_tol_to_gal}, the pair $\tu{{}_\sim,{}^\sim}$ is an extensive isotone Galois connection. The following theorem connects $\latU/{\sim}$ to the completely lattice $\lat$-ordered set $\Fix_{\tu{{}_\sim,{}^\sim}}$.
\begin{theorem}[factorization theorem]\label{thm:fact}
The $\lat$-ordered set $\latU/{\sim}$ is isomorphic to the completely lattice $\lat$-ordered set $\Fix_{\tu{{}_\sim,{}^\sim}}$ and, as such, is itself a completely lattice $\lat$-ordered set. The isomorphism is given by $\interval{v,u}\to \tu{u,v}$.
\end{theorem}
\begin{proof}
Follows directly from Lemma \ref{lem:ext_approx}, \ref{lem:ext_preceq} and definition of $\lat$-order on $\Fix_{\tu{{}_\sim,{}^\sim}}$.
\end{proof}

The second main result is that complete tolerances on completely lattice $\lat$-ordered sets can be represented by extensive isotone Galois connections.
\begin{theorem}[representation theorem]\label{thm:repre}
The mapping
\begin{align*}
{\sim} \mapsto \tu{{}_\sim,{}^\sim}
\end{align*}
%
%
%
is an isomorphism between $\ctol \latU$ and $\eigal(\latU)$.
Its inverse is
\begin{align*}
\tu{f,g} &\mapsto {\sim}_{\tu{f,g}}.
\end{align*}
$\ctol \latU$ and $\eigal(\latU)$ are both completely lattice $\lat$-ordered sets.

\end{theorem}
\begin{proof}
Follows from Theorem~\ref{thm:from_tol_to_gal}, Lemma~\ref{lem:tol_gal_tol}, Theorem~\ref{thm:from_gal_to_tol}, Lemma~\ref{lem:tol_to_gal_iso}, and Theorem~\ref{thm:ctol_clos}.
\end{proof}

\section{Conclusion}
We introduced a notion of complete binary fuzzy relation on complete fuzzy lattice (completely lattice fuzzy ordered set). The notion leads in ordinary (crisp) case to the classical notion of complete relation on complete lattice, but re-formulated in terms of the theory of power structures. We proved some basic properties of power structures of fuzzy ordered sets.

In the main part of the paper, we defined complete fuzzy binary relations and complete fuzzy tolerances and investigated their properties. Our main results are covered in Theorem \ref{thm:fact} and \ref{thm:repre}. We show that a fuzzy complete  lattice can be factorized by means of a complete fuzzy tolerance and that there is a naturally-defined structure of fuzzy complete lattice on the factor set. This result corresponds to the known result from the ordinary case \cite{Cze:Flt,Wille:Ctrcl}.

In addition, we found an isomorphism between the fuzzy ordered sets of all complete fuzzy tolerances and extensive isotone fuzzy Galois connections on a fuzzy complete lattice. This result is useful for testing fuzzy tolerances for completeness and is new even in the ordinary (crisp) setting. 

Our future research will focus on applying results from this paper to Formal Concept Analysis of data with fuzzy attributes \cite{Bel:FRS}. In ordinary setting, there is a correspondence between complete tolerances on a concept lattice and so called block relations of the associated formal context \cite{Wille:Ctrcl,GanWi:FCA}. Theorem \ref{thm:fact} and \ref{thm:repre} will help establish a link between complete fuzzy tolerances on a fuzzy concept lattice and (properly defined) block relations on the formal context. This will allow generalize results from \cite{Wille:Ctrcl,GanWi:FCA} to fuzzy concept lattices.

One of the consequences of our results is that the condition of compatibility from the definition of complete relation on a completely lattice $\lat$-ordered set (Sec.~\ref{Sec:compl}) is redundant for $\lat$-tolerances. This leads to an open problem, namely, whether the condition of compatibility follows from the other conditions of the definition.

%
%
%


\end{document}